\documentclass[a4paper,UKenglish,cleveref,autoref]{lipics-v2021}

\usepackage{tikz}
\usepackage{booktabs}
\usetikzlibrary{decorations.pathreplacing,arrows.meta,positioning}

\nolinenumbers

\title{Reportable Shortcuts with RLBWTs}

\author{Travis Gagie}{CeBiB -- Center for Biotechnology and Bioengineering, Chile}{travis.gagie@gmail.com}{https://orcid.org/0000-0003-3689-327X}{}

\authorrunning{T.~Gagie}

\Copyright{Travis Gagie}



\acknowledgements{We thank Ben Langmead for asking us, several years ago,
whether BWT-based indexes could get the kind of sequential-access speed
that $k$-mer indexes built on compacted de Bruijn graphs get once they
reach a unitig.}

\category{}
\relatedversion{}

\EventEditors{}
\EventNoEds{0}
\EventLongTitle{}
\EventShortTitle{}
\EventAcronym{}
\EventYear{}
\EventDate{}
\EventLocation{}
\EventLogo{}
\SeriesVolume{}
\ArticleNo{}

\newcommand{\Rref}{R}
\newcommand{\intv}{I}
\newcommand{\sym}{\mathit{sym}}

\begin{document}

\maketitle

\begin{abstract}
Backward search over a run-length compressed Burrows--Wheeler transform
(RLBWT) takes one step per pattern character, even when long stretches of
the pattern match material that recurs throughout a highly repetitive
collection. We describe \emph{shortcuts}: for a precomputed set of pairs
$(\alpha,\beta)$ of strings, a backward search that has just matched
$\beta$ and finds $\alpha$ immediately to its left in the pattern can
cross $\alpha$ in one step. Since a shortcut is taken only after its
pair has been recognized in the pattern, the search can record it, and
it returns a report of the shortcuts it took along with the usual BWT
interval. Our motivation is pangenome graphs, where such pairs mark the
places where a pattern passes through a vertex. We work out the case of
a single linear reference $\Rref$, in which $\beta$ ranges over the length-$\ell$
windows of $\Rref$ starting at multiples of $\ell$, and $\alpha$ over the
$2^t \ell$ characters of $\Rref$ preceding such a window. We show that
the resulting search counts correctly; that it crosses a stretch of the
pattern matching $m$ consecutive windows of $\Rref$, whose anchors
were all kept, with
$O(\log m + m / 2^\tau)$ shortcuts, where $2^\tau \ell$ is the longest
shortcut, provided the stretch and the rest of the pattern to its right
occur together in the collection; and that the shortcuts can be built from the RLBWT in
$O(|\Rref| (\tau + 1) \log r)$ time plus $O(\log n)$ time per BWT row
they cover, plus the time to sort the anchors. On the 461-megabyte Pizza\&Chili
\texttt{cere} collection of 37 genomes, shortcuts make backward
search $2.3$ to $4.2$ times faster on exact patterns of length 200 to
1000, and $1.7$ to $2.0$ times faster on patterns with 1\% of their
characters substituted, for 86\% more space than the RLBWT. On an
11-gigabase collection of genomes of seven bacterial species, with one
chromosome per species as reference, they make it $1.6$ to $1.9$ times
faster on exact patterns for 32\% more space; there, the speedup is
limited by how much the genomes differ from the reference. On whole
chromosome~19 haplotypes against 16 to 64 others from the Human
Pangenome Reference Consortium, shortcuts in both directions make
Boyer--Moore--Li search for maximal exact matches of length at least
1000 about 8 times faster than with the RLBWTs alone.
\end{abstract}

\clearpage
\section{Introduction}
\label{sec:intro}

The Burrows--Wheeler transform (BWT)~\cite{BurrowsWheeler1994} of a highly
repetitive text, such as a collection of genomes from one species, has
few runs of equal characters. Indexes built on its run-length compressed
form (RLBWT), such as the RLFM-index~\cite{MakinenNavarro2005} and the
$r$-index~\cite{GagieNavarroPrezza2020}, take space proportional to the
number $r$ of runs instead of to the length $n$ of the text; such
indexes are surveyed by Navarro~\cite{Navarro2021}. Counting the
occurrences of a pattern $P$ with such an index is done by \emph{backward
search}~\cite{FerraginaManzini2000}: $P$ is processed from right to left
while maintaining the BWT interval of the suffix of $P$ processed so far,
one backward step per character. The index is small, but the search
still takes one step per character even when the pattern matches
material that occurs many times in the collection. Indexes for $k$-mers
built on compacted de Bruijn graphs do better here: once a $k$-mer of
the pattern has been found in a unitig, the following $k$-mers can often
be checked just by comparing the next characters of the pattern with the
unitig~\cite{Pibiri2022}.

\Cref{fig:dbg} shows a collection of twelve strings and its de Bruijn
graph for $k = 3$, uncompacted and compacted. Every string contains the
unitig \texttt{AACGCTC}, and ten of them contain the string
\texttt{GTCAAGG}, which would be a unitig if two strings did not have
\texttt{GTCCAGG} instead; we call such a string an
\emph{almost-unitig}. The unitig and the almost-unitig appear in
different orders in different strings, so neither can be extended.
Every occurrence of a unitig's last $k$-mer is preceded by the rest of
the unitig, unless a string starts inside the unitig. In the BWT, this
means that the interval of the last $k$-mer is mapped one-to-one, in
order, onto the interval of the whole unitig, so a backward search that
has matched \texttt{CTC}, the last $3$-mer of \texttt{AACGCTC}, and then
finds \texttt{AACG} to its left in the pattern can cross those four
characters in one step. For an almost-unitig, some rows drop out on the
way, and a bitvector over the interval of the last $k$-mer says which.
With a reference, the same idea applies to the windows of the reference:
in \Cref{fig:dbg}, with the first string as the reference and windows of
length $4$ starting at multiples of $4$, a search that has matched the
window \texttt{GTGG} at position $16$ can check whether the pattern has
\texttt{CAAG} to its left and, if so, cross it in one step; of the six
rows in the BWT interval of \texttt{GTGG}, the one from the string with
\texttt{GTCCAGG} drops out. The shortcuts a search takes form a path
through the compacted de Bruijn graph or, with windows of a reference,
through a graph whose vertices are the windows; so a BWT-based index
with shortcuts can report the path a match takes through a pangenome
graph, and not only the BWT interval where the match ends up.

\begin{figure}[t]
  \centering
  \small
  \begin{tabular}{@{}l r@{}}
    \texttt{GTA\,AACGCTC\,GTCAAGG\,TGG} & $\times 5$ \\
    \texttt{GTA\,AACGCTC\,GTCCAGG\,TGG} & $\times 1$ \\
    \texttt{GTA\,GTCAAGG\,AACGCTC\,TGG} & $\times 5$ \\
    \texttt{GTA\,GTCCAGG\,AACGCTC\,TGG} & $\times 1$
  \end{tabular}

  \medskip
  \begin{tikzpicture}[x=0.62cm,y=0.62cm,>={Stealth[length=1.4mm]},v/.style={draw,rounded corners=1pt,inner sep=1pt,font=\tiny\ttfamily}]
\node[v,fill=gray!20] (GTA) at (0.00,0.00) {GTA};
\node[v,fill=gray!20] (AAC) at (3.00,1.60) {AAC};
\node[v,fill=gray!20] (CTC) at (7.00,1.60) {CTC};
\node[v,fill=gray!20] (GTC) at (3.00,-1.60) {GTC};
\node[v,fill=gray!20] (AGG) at (7.00,-1.60) {AGG};
\node[v,fill=gray!20] (TGG) at (10.00,0.00) {TGG};
\node[v] (TAA) at (1.00,0.53) {TAA};
\node[v] (AAA) at (2.00,1.07) {AAA};
\node[v] (TAG) at (1.00,-0.53) {TAG};
\node[v] (AGT) at (2.00,-1.07) {AGT};
\node[v] (ACG) at (4.00,1.60) {ACG};
\node[v] (CGC) at (5.00,1.60) {CGC};
\node[v] (GCT) at (6.00,1.60) {GCT};
\node[v] (TCG) at (5.67,0.53) {TCG};
\node[v] (CGT) at (4.33,-0.53) {CGT};
\node[v] (TCT) at (8.00,1.07) {TCT};
\node[v] (CTG) at (9.00,0.53) {CTG};
\node[v] (TCA) at (4.00,-1.60) {TCA};
\node[v] (CAA) at (5.00,-1.60) {CAA};
\node[v] (AAG) at (6.00,-1.60) {AAG};
\node[v] (TCC) at (4.00,-2.60) {TCC};
\node[v] (CCA) at (5.00,-2.60) {CCA};
\node[v] (CAG) at (6.00,-2.60) {CAG};
\node[v] (GGA) at (5.67,-0.53) {GGA};
\node[v] (GAA) at (4.33,0.53) {GAA};
\node[v] (GGT) at (8.00,-1.07) {GGT};
\node[v] (GTG) at (9.00,-0.53) {GTG};
\draw[->] (GTA) -- (TAA);
\draw[->] (GTA) -- (TAG);
\draw[->] (TAA) -- (AAA);
\draw[->] (AAA) -- (AAC);
\draw[->] (AAC) -- (ACG);
\draw[->] (ACG) -- (CGC);
\draw[->] (CGC) -- (GCT);
\draw[->] (GCT) -- (CTC);
\draw[->] (CTC) -- (TCG);
\draw[->] (CTC) -- (TCT);
\draw[->] (TCG) -- (CGT);
\draw[->] (CGT) -- (GTC);
\draw[->] (GTC) -- (TCA);
\draw[->] (GTC) -- (TCC);
\draw[->] (TCA) -- (CAA);
\draw[->] (CAA) -- (AAG);
\draw[->] (AAG) -- (AGG);
\draw[->] (AGG) -- (GGT);
\draw[->] (AGG) -- (GGA);
\draw[->] (GGT) -- (GTG);
\draw[->] (GTG) -- (TGG);
\draw[->] (TCC) -- (CCA);
\draw[->] (CCA) -- (CAG);
\draw[->] (CAG) -- (AGG);
\draw[->] (TAG) -- (AGT);
\draw[->] (AGT) -- (GTC);
\draw[->] (GGA) -- (GAA);
\draw[->] (GAA) -- (AAC);
\draw[->] (TCT) -- (CTG);
\draw[->] (CTG) -- (TGG);
\end{tikzpicture}

  \medskip
  \begin{tikzpicture}[x=0.62cm,y=0.62cm,>={Stealth[length=1.4mm]},v/.style={draw,rounded corners=1pt,inner sep=1.5pt,font=\scriptsize\ttfamily}]
\node[v] (GTA) at (0,0) {GTA};
\node[v] (TAAA) at (1.5,0.8) {TAAA};
\node[v] (TAGT) at (1.5,-0.8) {TAGT};
\node[v,fill=gray!20] (AACGCTC) at (5,1.6) {AACGCTC};
\node[v] (TCGT) at (5.6,0.45) {TCGT};
\node[v] (GGAA) at (4.4,0.45) {GGAA};
\node[v,fill=gray!8] (GTC) at (3,-1.6) {GTC};
\node[v,fill=gray!8] (TCAAG) at (5,-1.6) {TCAAG};
\node[v,fill=gray!8] (TCCAG) at (5,-2.6) {TCCAG};
\node[v,fill=gray!8] (AGG) at (7,-1.6) {AGG};
\node[v] (TCTG) at (8.5,0.8) {TCTG};
\node[v] (GGTG) at (8.5,-0.8) {GGTG};
\node[v] (TGG) at (10,0) {TGG};
\draw[->] (GTA) -- (TAAA);
\draw[->] (GTA) -- (TAGT);
\draw[->] (TAAA) -- (AACGCTC);
\draw[->] (TAGT) -- (GTC);
\draw[->] (AACGCTC) -- (TCGT);
\draw[->] (AACGCTC) -- (TCTG);
\draw[->] (TCGT) -- (GTC);
\draw[->] (GTC) -- node[font=\tiny,fill=white,inner sep=0.5pt] {10} (TCAAG);
\draw[->] (GTC) -- node[font=\tiny,fill=white,inner sep=0.5pt] {2} (TCCAG);
\draw[->] (TCAAG) -- node[font=\tiny,fill=white,inner sep=0.5pt] {10} (AGG);
\draw[->] (TCCAG) -- node[font=\tiny,fill=white,inner sep=0.5pt] {2} (AGG);
\draw[->] (AGG) -- (GGAA);
\draw[->] (AGG) -- (GGTG);
\draw[->] (GGAA) -- (AACGCTC);
\draw[->] (TCTG) -- (TGG);
\draw[->] (GGTG) -- (TGG);
\end{tikzpicture}
  \caption{Twelve strings (top; spaces added for readability) and their
  de Bruijn graph for $k = 3$, uncompacted (middle; branching $k$-mers
  shaded) and compacted (bottom). \texttt{AACGCTC} is a unitig.
  \texttt{GTCAAGG} is an almost-unitig: ten strings contain it, and two
  have \texttt{GTCCAGG} instead (edge labels count the strings).}
  \label{fig:dbg}
\end{figure}

We add \emph{shortcuts} to backward search and observe that a search
that takes them can report which ones it took. The starting
point is a simple observation (\Cref{lem:detect}): if a string $s$ is no
longer than the suffix $X$ of $P$ matched so far, then the BWT interval
of $X$ lies inside that of $s$ if and only if $s$ is a prefix of $X$.
Given a set of strings with precomputed BWT intervals, a backward search
can therefore notice, at every step, that one of them has just been
matched at the left end of its current match. Suppose we also know a
string $\alpha$ that commonly precedes such a string $\beta$ in the
collection. If the pattern has $\alpha$ immediately to the left of
$\beta$, which we can check by reading the pattern directly, we can cross
$\alpha$ in a single step, using precomputed information about which
rows in $\beta$'s interval are preceded by $\alpha$ (\Cref{lem:jump}). A
shortcut is taken only after the pair $(\alpha,\beta)$ has been
confirmed in the pattern, so alongside the final interval the search can
return the list of pairs it shortcut through.

Our motivation is pangenome graphs. For a vertex $v$, a set $A$ of
strings labelling the ends of paths into $v$ and a set $B$ of strings
labelling the beginnings of paths out of $v$ can be chosen, and a pair
$(\alpha,\beta) \in A \times B$ found in a pattern marks a place where
the pattern passes through $v$. A search for the longest suffix of a
pattern that occurs in the collection, taking shortcuts for selected
pairs, then returns that suffix's interval together with the vertices
it crossed by shortcuts (\Cref{sec:mechanism}). BWT-based indexes for graphs are
surveyed in~\cite{Prezza2021,CotumaccioEtAl2025,AlankoEtAl2025}. Here
we work out only the simplest case: a single linear reference
$\Rref$, one of the sequences in the collection
(\Cref{sec:linear}). There, $\beta$ ranges over the windows of length
$\ell$ of $\Rref$ starting at multiples of $\ell$, and $\alpha$ over the
$2^t \ell$ characters of $\Rref$ immediately preceding a window, for $0
\le t \le \tau$. Each such $\alpha$ is determined by its $\beta$ and its
length, so confirming it needs no hash table, only a comparison with
$\Rref$; and a shortcut always ends at the start of another window, so
shortcuts chain.

Two recent lines of work also attach extra information to BWT-based
searches. Tag arrays~\cite{BalazEtAl2024,EskandarEtAl2025} associate
BWT positions with locations in a pangenome graph, so that the
locations of a match can be listed from its interval. The
col-BWT~\cite{BrownEtAl2026} splits BWT runs along tunnels induced by
multi-MUMs, maximal strings that occur exactly once in every genome of
the collection~\cite{ShivakumarLangmead2025}, and labels the resulting
sub-runs, so that matching statistics come with chain statistics saying
which multi-MUM each match lies in. Tag arrays and the col-BWT annotate
the rows a search visits without changing the number of steps it takes;
shortcuts reduce the number of steps, and the report lists the pairs
crossed. The two approaches could meet: whole
multi-MUMs~\cite{ShivakumarLangmead2025,ShivakumarLangmead2026}, used as
windows of varying length, have pairwise disjoint intervals and so never
collide, and the report could then be given in the pangenome coordinates
that multi-MUMs define~\cite{ShivakumarLangmead2026Shredtools} instead
of in the coordinates of one reference.

Crossing many characters in one step and checking them afterwards is
also how CDAWG-based indexes search~\cite{BlumerEtAl1987,Navarro2021}.
An edge of the CDAWG leaving the node for a string $x$ with label $w$
means that every occurrence of $x\,w[0]$ in the text continues with the
rest of $w$, so the search follows the edge in one step and
verifies the skipped characters by extracting them from the text;
Belazzougui et al.~\cite{BelazzouguiEtAl2015,BelazzouguiEtAl2017}
combine the CDAWG with an RLBWT, which tells whether the pattern occurs,
so that no extraction is needed. Roughly, in our terms, an edge of the
CDAWG of the reversed text corresponds to a pair $(\alpha,\beta)$ such
that every row of $\intv(\beta)$ is preceded by $\alpha$. Shortcuts also
allow pairs for which only some of the rows are preceded by $\alpha$, and
we choose them from a reference, so their number is set by parameters
rather than by the size of the CDAWG, which can be much larger than
$r$~\cite{Navarro2021a}. The GBWT~\cite{SirenEtAl2020}
indexes haplotypes as paths in a pangenome graph and searches for
sequences of vertices, so each of its steps crosses a whole vertex, but
its patterns must be given as paths in the graph rather than as strings.

Another line of work searches forward, by comparing the pattern with
the text. A suffixient set~\cite{DepuydtEtAl2023} is a set of text
positions, at most twice as many as there are runs in the BWT of the
reversed text, that suffices, with random access to the text, to find
the maximal exact matches (MEMs) of a pattern; suffixient
arrays~\cite{CenzatoEtAl2024} store such positions in suffix-array order
and find MEMs by binary search, and indexes based on path decompositions
of the suffix tree~\cite{BeckerEtAl2026} sample at most $r$ suffix-array
entries and support locating with I/O complexity close to that of an
uncompressed suffix tree. These indexes are efficient in the I/O model,
and suffixient arrays are reported to be one to two orders of magnitude
faster than the $r$-index~\cite{CenzatoEtAl2024}. None of these
indexes, however, gives the BWT interval of a match, and the interval is useful for more
than counting. With a tag
array~\cite{BalazEtAl2024,Gagie2024tags}, for example, the occurrences of
a match in each taxon of a reference database can be counted by scanning
the tag runs that overlap the match's interval, without locating any
occurrence, which is how Gagie and Navarro classify
reads~\cite{GagieNavarro2026}. Boyer--Moore--Li~\cite{Gagie2024BML,Li2024}
finds long MEMs with backward searches, whose intervals can then be used
with a tag array, and shortcuts could speed up those searches. Shortcuts keep backward search and the
BWT interval, and obtain some sequential access by comparing the pattern
with a reference stored in plain form. For this reason we do not
compare shortcuts with suffixient arrays or with indexes based on
suffix-tree path decompositions in our experiments.

Our contributions are:
\begin{itemize}
\item the general mechanism: detecting strings by interval containment,
  shortcutting across confirmed pairs, and reporting the shortcuts
  taken (\Cref{sec:mechanism});
\item for a linear reference, a data structure that adds shortcuts to an
  RLBWT (\Cref{thm:space}), with a search that returns the same result
  as plain backward search (\Cref{thm:query});
\item a bound of $O(\log m + m / 2^\tau)$ shortcuts for crossing a
  stretch of $m$ consecutive windows of the reference (\Cref{lem:count});
\item a construction from the RLBWT that tracks where the rows of each
  entrance stop matching the reference, which happens only at run
  boundaries (\Cref{thm:construction});
\item experiments on the Pizza\&Chili \texttt{cere}
  collection~\cite{FerraginaEtAl2009} and on an 11-gigabase collection of
  bacterial genomes (\Cref{sec:experiments});
\item an application to finding maximal exact matches in whole
  chromosomes, with shortcuts in both directions (\Cref{sec:chr19}).
\end{itemize}

\section{Preliminaries}
\label{sec:prelim}

Let $T[0..n)$ be a text over an ordered alphabet $\Sigma$ of size
$\sigma$, ending with a
sentinel \$ that occurs nowhere else and is lexicographically smallest.
The suffix array $\mathrm{SA}[0..n)$ lists the starting positions of the
suffixes of $T$ in lexicographic order, and the BWT is
$\mathrm{BWT}[i] = T[(\mathrm{SA}[i] - 1) \bmod n]$. We call $i$ a
\emph{row}. For a string $X$, let $\intv(X)$ be the interval of rows
whose suffixes start with $X$; it may be empty. The LF-mapping sends row
$i$ to the row of the suffix one position to the left,
$\mathrm{LF}(i) = C[c] + \mathrm{rank}_c(\mathrm{BWT}, i)$ with $c =
\mathrm{BWT}[i]$, where $C[c]$ counts the characters in $T$ smaller than
$c$ and $\mathrm{rank}_c(\mathrm{BWT}, i)$ counts the occurrences of $c$
in $\mathrm{BWT}[0..i)$. A \emph{backward step} computes $\intv(cX) =
[C[c] + \mathrm{rank}_c(\mathrm{BWT}, \mathrm{lo}),\ C[c] +
\mathrm{rank}_c(\mathrm{BWT}, \mathrm{hi} + 1) - 1]$ from $\intv(X) =
[\mathrm{lo}, \mathrm{hi}]$, and \emph{backward search} for a pattern
$P$ applies backward steps for $P$'s characters from right to left. We
use it to find the longest suffix of $P$ that occurs in $T$, together
with that suffix's interval, whose size is its number of occurrences.

The RLBWT stores the BWT as $r$ runs of equal characters. With standard
representations~\cite{MakinenNavarro2005,GagieNavarroPrezza2020}, it
takes $O(r)$ words, a backward step takes $O(\log r)$ time, and the runs
intersecting an interval can be listed in $O(\log r)$ time plus constant
time per run. We build the RLBWT with prefix-free
parsing~\cite{BoucherEtAl2019}, which avoids holding the full suffix
array in memory.

We use two standard facts. First, for strings $X \ne Y$ of the same
length, $\intv(X)$ and $\intv(Y)$ are disjoint. Second, LF-mapping
preserves the order of rows with equal BWT characters, so for any string
$\alpha$, applying LF $|\alpha|$ times to the rows whose suffixes are
preceded by $\alpha$ in $T$ preserves their order.

\section{Shortcuts and their reports}
\label{sec:mechanism}

\subsection{Detecting strings during backward search}
\label{sec:detect}

\begin{lemma}
\label{lem:detect}
Let $X$ and $s$ be strings with $|X| \ge |s|$ and $\intv(X)$ non-empty.
Then $\intv(X) \subseteq \intv(s)$ if and only if $s$ is a prefix of
$X$.
\end{lemma}

\begin{proof}
If $s$ is a prefix of $X$, every suffix starting with $X$ starts with
$s$. Conversely, a row in $\intv(X) \subseteq \intv(s)$ has a suffix
starting with both $X$ and $s$, and since $|X| \ge |s|$, $s$ is a prefix
of $X$.
\end{proof}

The length condition is needed: if $X$ is a proper prefix of $s$ and
every occurrence of $X$ in $T$ extends to an occurrence of $s$, then
$\intv(X) = \intv(s)$. A backward search knows the length of its current
match, so the condition is easy to check.

Fix a set $S = \{s_1, \ldots, s_k\}$ of strings occurring in $T$ and
store their intervals. During backward search for $P$, after each step,
if the current interval $\intv(X)$ is contained in some $\intv(s_i)$
with $|s_i| \le |X|$, then $s_i$ is a prefix of $X$: it has just been
matched at the left end of the current match. If the intervals of
strings in $S$ are pairwise disjoint, as for strings of equal length, a
predecessor query on their endpoints finds the only candidate. Otherwise
the intervals containing $\intv(X)$ can be found with an interval tree,
or by two-sided range reporting on the intervals' endpoints.

\subsection{Vertex crossings}
\label{sec:crossings}

Consider a pangenome graph and a vertex $v$. Let $A$ be a set of strings
labelling the ends of paths into $v$, and $B$ a set of strings
labelling the beginnings of paths out of $v$. (Which paths to use is a
question about the graph that we do not address here; one option is to
use the paths taken by the haplotypes stored in a
GBWT~\cite{SirenEtAl2020}.) A
\emph{crossing} of $v$ in a pattern $P$ is an occurrence of
$\alpha\beta$ in $P$ with $\alpha \in A$ and $\beta \in B$.

Since backward search processes $P$ from right to left, it matches
$\beta$ before $\alpha$. Setting $S = B$ in \Cref{sec:detect}, the
search detects $\beta$ as soon as $\beta$ has been matched at the left
end of its current match. Whether $\beta$ is immediately preceded in $P$
by some $\alpha \in A$ is then checked by reading $P$ further to the
left and looking up the strings read in a hash table storing $A$; the
strings of $A$ need not be indexed.

\begin{figure}[t]
  \centering
  \begin{tikzpicture}[x=0.62cm, y=0.62cm, >={Stealth[length=2.2mm]}]
    \foreach [count=\i] \c in {G,A,T,T,A,C,A,C,A,G,T} {
      \node[draw, minimum width=0.62cm, minimum height=0.62cm, inner sep=0pt]
        (c\i) at (\i,0) {\texttt{\c}};
    }
    \draw[decorate, decoration={brace, amplitude=4pt, raise=2pt}]
      (0.72,0.4) -- (7.28,0.4) node[midway, above=6pt] {$\alpha = \texttt{GATTACA}$};
    \draw[decorate, decoration={brace, amplitude=4pt, raise=2pt}]
      (7.72,0.4) -- (11.28,0.4) node[midway, above=6pt] {$\beta = \texttt{CAGT}$};
    \draw[->] (11.3,-0.9) .. controls (10.2,-1.5) and (8.8,-1.5) .. (7.7,-0.9)
      node[midway, below=2pt] {\small 4 steps};
    \draw[->, thick] (7.3,-0.9) .. controls (5.5,-2.0) and (2.5,-2.0) .. (0.7,-0.9)
      node[midway, below=2pt] {\small one shortcut};
  \end{tikzpicture}
  \caption{A crossing in the pattern \texttt{GATTACACAGT}. Backward search
    matches $\beta = \texttt{CAGT}$ in four steps and detects it by
    interval containment (\Cref{lem:detect}). Reading the pattern to the
    left confirms $\alpha = \texttt{GATTACA}$, and if $(\alpha,\beta)$ is
    one of the stored pairs, the search crosses $\alpha$ with one
    shortcut instead of seven steps (\Cref{ex:crossing}).}
  \label{fig:crossing}
\end{figure}

\begin{example}
\label{ex:crossing}
Let $\texttt{GATTACA} \in A$ and $\texttt{CAGT} \in B$, and let the
pattern be $P = \texttt{GATTACACAGT}$. After four steps, backward search
has matched $X = \texttt{CAGT}$ and its interval is $\intv(\beta)$, so
$\beta$ is detected. Reading $P$ further to the left gives
$\texttt{GATTACA} \in A$, so $P$ contains a crossing of $v$ there
(\Cref{fig:crossing}). For $P = \texttt{TTACACAGT}$, $\beta$ is detected
in the same way, but the seven characters needed to confirm
$\texttt{GATTACA}$ are not there, so this crossing is not found.
\end{example}

To report every crossing in the longest suffix of $P$ occurring in $T$,
the search must still process that whole suffix one step at a time.

\subsection{Shortcuts}
\label{sec:shortcuts}

Suppose instead that we want the interval of the longest suffix of $P$
occurring in $T$ and a set of crossings in it. Then crossings can speed
up the search that finds them. Fix a set $\mathcal{C}$ of pairs
$(\alpha,\beta)$ with $\intv(\alpha\beta)$ non-empty, chosen in advance
(for example, pairs that occur often in $T$). For each, store the first
row $f_{\alpha\beta}$ of
$\intv(\alpha\beta)$ and a bitvector $M_{\alpha\beta}$ over
$\intv(\beta)$ marking the rows whose suffixes are preceded by $\alpha$
in $T$. For a row $i$ in or just after $\intv(\beta)$, let
$\mathrm{rank}_{\alpha\beta}(i)$ be the number of marked rows of
$\intv(\beta)$ before $i$.

\begin{lemma}
\label{lem:jump}
Let $\beta$ be a prefix of $X$, with $\intv(X) = [\mathrm{lo},
\mathrm{hi}] \subseteq \intv(\beta)$ non-empty, let $\intv(\alpha\beta)$ be
non-empty, and let $\alpha$ not contain \$. Then
\[
\intv(\alpha X) = [\,f_{\alpha\beta} + \mathrm{rank}_{\alpha\beta}(\mathrm{lo}),\
f_{\alpha\beta} + \mathrm{rank}_{\alpha\beta}(\mathrm{hi}+1) - 1\,],
\]
which is empty if and only if no row of $\intv(X)$ is marked.
\end{lemma}

\begin{proof}
Applying LF $|\alpha|$ times maps the marked rows of $\intv(\beta)$ one
to one onto the rows whose suffixes start with $\alpha\beta$, that is,
onto $\intv(\alpha\beta)$, and preserves their order (\Cref{sec:prelim}).
The marked rows of $\intv(X)$ are those whose suffixes are preceded by
$\alpha$, and they are mapped onto the rows whose suffixes start with
$\alpha X$. Since $\intv(X)$ is an interval, these marked rows are
consecutive among the marked rows of $\intv(\beta)$, with
$\mathrm{rank}_{\alpha\beta}(\mathrm{lo})$ marked rows before them, so
their images are the stated rows of $\intv(\alpha\beta)$.
\end{proof}

\begin{figure}[t]
  \centering
  \begin{tikzpicture}[x=1cm, y=0.5cm, >={Stealth[length=2mm]}]
    \foreach \i/\m in {0/0,1/1,2/0,3/1,4/1,5/0,6/1,7/0} {
      \ifnum\m=1
        \node[draw, fill=gray!35, minimum width=0.7cm, minimum height=0.5cm, inner sep=0pt] (L\i) at (0,-\i) {\scriptsize $\bullet$};
      \else
        \node[draw, minimum width=0.7cm, minimum height=0.5cm, inner sep=0pt] (L\i) at (0,-\i) {};
      \fi
    }
    \foreach \i in {0,...,5} {
      \node[draw, minimum width=0.7cm, minimum height=0.5cm, inner sep=0pt] (R\i) at (5,-\i-1) {};
    }
    \foreach \i in {1,...,4} { \node[draw, fill=gray!35, minimum width=0.7cm, minimum height=0.5cm, inner sep=0pt] at (5,-\i-1) {}; }
    \draw[->] (L1.east) -- (R1.west);
    \draw[->] (L3.east) -- (R2.west);
    \draw[->] (L4.east) -- (R3.west);
    \draw[->] (L6.east) -- (R4.west);
    \draw[decorate, decoration={brace, amplitude=4pt, mirror}] (-0.45,-2.75) -- (-0.45,-6.25)
      node[midway, left=6pt] {$\intv(X)$};
    \draw[decorate, decoration={brace, amplitude=4pt}] (5.45,-2.75) -- (5.45,-5.25)
      node[midway, right=6pt] {$\intv(\alpha X)$};
    \node[above] at (0,0.3) {$\intv(\beta)$};
    \node[above] at (5,-0.7) {$\intv(\alpha\beta)$ and neighbours};
    \node[right] at (5.4,-2) {\scriptsize $\leftarrow f_{\alpha\beta}$};
  \end{tikzpicture}
  \caption{\Cref{lem:jump}. The marked rows of $\intv(\beta)$ (shaded) are
    those preceded by $\alpha$; applying LF $|\alpha|$ times maps them in order onto
    $\intv(\alpha\beta)$. Here $\intv(X)$ consists of the fourth to seventh
    rows of $\intv(\beta)$, one marked row precedes it and three of its
    rows are marked, so $\intv(\alpha X)$ consists of the second to fourth
    rows of $\intv(\alpha\beta)$. The marked rows of $\intv(\beta)$ need not
    be consecutive, which is why a rank structure is needed.}
  \label{fig:jump}
\end{figure}

The search proceeds as usual, except that whenever it detects some
$\beta$ with $(\alpha,\beta) \in \mathcal{C}$ and confirms $\alpha$ to
the left in $P$, it computes $\intv(\alpha X)$ by \Cref{lem:jump} instead
of taking $|\alpha|$ backward steps, provided the result is non-empty
(\Cref{fig:jump}). It records the pair and the position in $P$ where the
shortcut was taken. The search returns the same interval as plain
backward search, since a non-empty $\intv(\alpha X)$ means all the
intermediate intervals are non-empty too. Its \emph{report} is the list
of recorded shortcuts. Producing it takes constant time per shortcut,
because the pair has already been identified when the shortcut is taken.

\Cref{lem:jump} computes the same interval as a lemma of Huynh, Hon,
Lam and Sung~\cite{HuynhEtAl2006}, also used for indexed multi-pattern
matching~\cite{GagieEtAl2012}, which obtains $\intv(P_1 P_2)$ from
$\intv(P_1)$ and $\intv(P_2)$ by binary search over $\intv(P_1)$, with
access to the suffix array and its inverse. Because $\alpha$ and $\beta$
are fixed in advance, we can precompute which rows of $\intv(\beta)$
are preceded by $\alpha$, so a shortcut needs only rank queries and no
access to the suffix array or its inverse, which an RLBWT does not
support quickly.

\begin{example}
\label{ex:shortcut}
Continuing \Cref{ex:crossing}, suppose $(\texttt{GATTACA},
\texttt{CAGT}) \in \mathcal{C}$. Searching for $P =
\texttt{GATTACACAGT}$, after four steps the search has $\intv(X) =
\intv(\beta)$ with $X = \beta$, confirms $\alpha$ in $P$, and crosses it
with one shortcut, landing on $\intv(\texttt{GATTACACAGT})$. It reports
the crossing of $v$ at the start of $P$.
\end{example}

\section{The linear-reference case}
\label{sec:linear}

From now on the graph is a single path: a reference $\Rref$, which is one
of the sequences in $T$ (or a prefix of one). Fix a window length $\ell$
and a maximum tier $\tau$. The \emph{anchors} are the positions $p \in
\{\ell, 2\ell, \ldots\}$ with $p + \ell \le |\Rref|$. Anchor $p$ has
window $\beta_p = \Rref[p..p+\ell)$, maximum tier
$\tau_p = \max\{t \le \tau : 2^t \ell \text{ divides } p\}$, and, for
$0 \le t \le \tau_p$, tier-$t$ string $\alpha_{p,t} = \Rref[p - 2^t\ell
.. p)$. The pairs $(\alpha_{p,t}, \beta_p)$ play the role of
$\mathcal{C}$.

If two anchors have the same window, a detected window does not say which
anchor's tier strings to compare against. We discard such anchors: they
are found by comparing intervals, since two windows are equal if and only
if their intervals are (\Cref{sec:prelim}). Recovering them, as one
$\beta$ with several candidate $\alpha$'s told apart with a hash table as
in \Cref{sec:crossings}, is left for future work. The \emph{entrance} of
a remaining anchor $p$ is $\intv(\beta_p)$; entrances are pairwise
disjoint.

The linear case is simpler than the general one in three ways.
\begin{itemize}
\item \emph{Confirming $\alpha$.} The tier strings of $p$ are determined
  by $p$, so confirming $\alpha_{p,t}$ means comparing the characters of
  $P$ to the left of the current match with $\Rref[p-1], \Rref[p-2],
  \ldots$; no hash table is needed.
\item \emph{One symbol per row.} Since $\alpha_{p,t-1}$ is a suffix of
  $\alpha_{p,t}$, a row of the entrance preceded by $\alpha_{p,t}$ is
  also preceded by $\alpha_{p,t-1}$. Instead of one bitvector $M$ per
  tier we can store, for each row $i$ of the entrance, the number
  $\sym(i) \in [0, \tau_p + 1]$ of tiers whose strings precede it
  (\Cref{fig:tiers}). We say row $i$ \emph{survives} tier $t$ if
  $\sym(i) > t$.
\item \emph{Chaining.} The string $\alpha_{p,t}\beta_p = \Rref[p - 2^t
  \ell .. p + \ell)$ starts with $\beta_{p'}$ for $p' = p - 2^t\ell$, and
  $p'$ is a multiple of $2^t \ell$. So a tier-$t$ shortcut ends in the
  entrance of $p'$ and, unless $p' = 0$ or $p'$ was discarded, dropped
  or capped (\Cref{sec:design}), $p'$ offers tier-$t$ shortcuts itself.
\end{itemize}

\begin{figure}[t]
  \centering
  \begin{tikzpicture}[x=1.0cm, y=0.62cm]
    \foreach [count=\i] \lab in {4,3,2,1} {
      \node[draw, minimum width=1.5cm, minimum height=0.62cm, inner sep=0pt]
        at (\i*1.6,0) {\scriptsize $\beta_{p-\lab\ell}$};
    }
    \node[draw, fill=gray!25, minimum width=1.5cm, minimum height=0.62cm, inner sep=0pt]
      at (8,0) {$\beta_p$};
    \draw[decorate, decoration={brace, amplitude=3pt, raise=2pt}]
      (5.65,0.35) -- (7.15,0.35) node[midway, above=5pt] {\scriptsize $\alpha_{p,0}$};
    \draw[decorate, decoration={brace, amplitude=3pt, raise=14pt}]
      (4.05,0.35) -- (7.15,0.35) node[midway, above=17pt] {\scriptsize $\alpha_{p,1}$};
    \draw[decorate, decoration={brace, amplitude=3pt, raise=26pt}]
      (0.85,0.35) -- (7.15,0.35) node[midway, above=29pt] {\scriptsize $\alpha_{p,2}$};
    \draw[|-|] (7.5,-0.6) -- (8.75,-0.6) node[midway, below] {\scriptsize $\ell$};
    \draw (7.2,-0.35) -- (7.2,-0.6);
    \node[below] at (7.2,-0.6) {\scriptsize $p$};
  \end{tikzpicture}
  \caption{The tier strings of an anchor $p$ with $\tau_p = 2$; each box
    is a window of $\ell$ characters of $\Rref$. The tier-$t$ string covers
    the $2^t$ windows left of $p$, and a tier-$t$ shortcut ends at
    $p - 2^t \ell$, the anchor of the leftmost window it crossed. A row of the
    entrance $\intv(\beta_p)$ preceded by $\alpha_{p,2}$ is also preceded by
    $\alpha_{p,1}$ and $\alpha_{p,0}$, so one symbol per row records which
    tiers it survives.}
  \label{fig:tiers}
\end{figure}

\paragraph*{Searching.}
The search keeps $\intv(X)$ and $|X|$ for the current match $X =
P[j..|P|)$. After each step, if $|X| \ge \ell$ and $\intv(X)$ lies in the
entrance of some anchor $p$, then $\beta_p$ is a prefix of $X$
(\Cref{lem:detect}). The search compares $P[j-1], P[j-2], \ldots$ with
$\Rref[p-1], \Rref[p-2], \ldots$ for up to $2^{\tau_p}\ell$ characters,
and lets $t$ be the largest tier with $\alpha_{p,t}$ fully matched. It
computes $\intv(\alpha_{p,t} X)$ by \Cref{lem:jump}, with the rows of
symbol greater than $t$ as the marked rows; if this is empty, it tries
the next smaller tier, and if all are empty, it takes a backward step.
A tier-$t$ shortcut is recorded as the triple $(j, p, t)$ and certifies
that $P[j - 2^t\ell .. j + \ell) = \Rref[p - 2^t\ell .. p + \ell)$.
So the report is a set of exact alignments of pieces of $P$ to $\Rref$;
the pieces of two chained shortcuts, the second taken right where the
first ends, overlap by exactly one window.

\paragraph*{Number of shortcuts.}
Shortcuts can only start at multiples of $2^t\ell$, so crossing $m$
windows may need more than the $O(\log m)$ shortcuts suggested by
writing $m$ as a sum of powers of two. The next lemma shows it needs at
most about twice that, plus $\lfloor m / 2^\tau \rfloor$ shortcuts of the
longest length.

\begin{lemma}
\label{lem:count}
Suppose $P[j - m\ell .. j + \ell) = \Rref[p - m\ell .. p + \ell)$, the
current match $X = P[j..|P|)$ lies in the entrance of anchor $p$,
$P[j - m\ell .. |P|)$ occurs in $T$, and none of the anchors $p - (m-1)\ell,
\ldots, p$ was discarded, dropped or capped. Then the search reaches or
passes position $j - m\ell$ of $P$ with at most $2\lceil \log_2(m+1) \rceil +
\lfloor m / 2^\tau \rfloor$ shortcuts.
\end{lemma}

\begin{proof}
Measure positions in units of $\ell$: the search starts at $x_0 =
p/\ell$ and must reach $a = x_0 - m$. From $x > a$, the tier
$t = \min(\nu(x), \tau, \lfloor \log_2(x - a) \rfloor)$ is usable, where
$\nu(x)$ is the exponent of the largest power of two dividing $x$: its
string is matched by $P$, and its interval is non-empty because
$P[j - m\ell .. |P|)$ occurs in $T$. The search takes this tier or a
larger one. A larger one has $2^t > x - a$, so it passes $a$ and is
the last shortcut needed; so it suffices to bound the number of shortcuts
when the search always takes the stated tier, which leads to $x - 2^t$. While $t = \nu(x)$ is smaller than the
other two terms, subtracting $2^{\nu(x)}$ increases $\nu$, so these
shortcuts use strictly increasing tiers below $\lfloor \log_2 m
\rfloor + 1$: at most $\lceil \log_2(m+1) \rceil$ of them. While
$t = \tau \le \nu(x)$, each shortcut crosses $2^\tau$ windows and keeps
$\nu \ge \tau$: at most $\lfloor m / 2^\tau \rfloor$ of them. Once $t =
\lfloor \log_2(x - a) \rfloor \le \nu(x)$, the remaining distance drops
below $2^t$ while $\nu$ stays at least $t$, so later shortcuts use
strictly decreasing tiers: at most $\lceil \log_2(m+1) \rceil$ of them.
The phases occur in this order, which gives the bound.
\end{proof}

\paragraph*{Relation to tunnelling.}
Baier~\cite{Baier2018} introduced \emph{tunnelling}, which merges
parallel, equally labelled paths in the BWT, viewed as a graph, to
compress it. Choosing the best tunnels is NP-hard when they may
overlap~\cite{BaierDede2019}; tunnelling was extended to Wheeler
graphs~\cite{AlankoEtAl2019} and has been built for large collections with
prefix-free parsing~\cite{GogaBalaz2022}, and the col-BWT uses tunnels to
label runs~\cite{BrownEtAl2026}. Our shortcuts have a similar flavour to
Baier's tunnels, in that both exploit sets of rows whose preceding
contexts agree for a long stretch, but we use them for a different
purpose: to speed up backward search and to report where a pattern
aligns with the reference, not to compress the BWT. One technical
difference is that a tunnel requires all its rows to continue together,
whereas a shortcut lets any subset of the rows of its entrance continue,
and the rank queries of \Cref{lem:jump} say which.

\section{The data structure}
\label{sec:design}

Let $E$ be the number of entrances and $W$ the total number of rows they
contain. Since entrances are disjoint, $W \le n$. We store:
\begin{itemize}
\item the reference $\Rref$, in $|\Rref| \lceil \log_2 \sigma \rceil$
  bits;
\item a sparse bitvector~\cite{OkanoharaSadakane2007} marking the first
  row of each entrance, in $O(E \log(n/E))$ bits, and the last row of each
  entrance; for $\intv(X) = [\mathrm{lo}, \mathrm{hi}]$, one rank query
  gives the only entrance that can contain it, and comparing $\mathrm{hi}$
  with that entrance's last row decides whether it does;
\item for each entrance, its anchor $p$, the position of its first row
  in the concatenation of all entrances, and the first row
  $f_{p,t}$ of $\intv(\alpha_{p,t}\beta_p)$ for each tier $t \le
  \tau_p$, in $O(\min\{E (\tau + 1), |\Rref| / \ell\} \log n)$ bits in
  total (over all anchors, the sum of $\tau_p + 1$ is $O(|\Rref| / \ell)$,
  since half the anchors have $\tau_p = 0$, a quarter have $\tau_p = 1$,
  and so on);
\item the sequence of symbols $\sym(i)$ of all rows of all entrances, in
  row order, in a wavelet tree~\cite{GrossiGuptaVitter2003} over the
  alphabet $[0, \tau + 1]$.
\end{itemize}
The number of rows of an entrance before row $i$ with symbol greater than
$t$, which \Cref{lem:jump} needs, is a range count on the wavelet tree,
in $O(\log \tau)$ time. Storing $\tau + 1$ cumulative bitvectors ``symbol
greater than $t$'' would answer it with one rank query, but takes $W
(\tau + 1)$ bits instead of $W \lceil \log_2(\tau + 2) \rceil$; we use
the wavelet tree for its space.

\begin{theorem}
\label{thm:space}
The shortcuts take $W \lceil \log_2(\tau + 2) \rceil (1 + o(1)) +
O(\min\{E (\tau + 1), |\Rref| / \ell\} \log n) + |\Rref| \lceil \log_2
\sigma \rceil$ bits on top of the RLBWT.
\end{theorem}

\begin{theorem}
\label{thm:query}
The search of \Cref{sec:linear} finds the same longest suffix of $P$
occurring in $T$, and the same interval, as plain backward search. Each
backward step takes $O(\log r)$ time plus one rank query on the sparse bitvector,
and each attempt to take a shortcut at an anchor $p$ compares at most
$2^{\tau_p} \ell$ characters of $P$ and $\Rref$ and takes
$O((\tau + 1) \log(\tau + 2))$ further time.
\end{theorem}

\begin{proof}
By \Cref{lem:detect}, the search only treats $\intv(X)$ as lying in an
entrance when $\beta_p$ is a prefix of $X$. By \Cref{lem:jump}, a
shortcut produces $\intv(\alpha_{p,t}X)$, the interval plain search
reaches after $2^t\ell$ more steps, and it is taken only if this interval
is non-empty, in which case all intervals plain search passes through are
non-empty too. If every tier gives an empty interval, the search falls
back to backward steps. The time bounds follow from the representations
above; trying tiers from the largest down costs $O(\log(\tau + 2))$ per
tier.
\end{proof}

In the worst case an attempt can fail at every step, after comparing up
to $2^\tau \ell$ characters each time, so the comparisons for a pattern
are bounded only by $O(|P| \, 2^\tau \ell)$. An attempt fails only when
$P$ agrees with $\Rref$ to the left of the window for fewer than $\ell$
characters, or when no occurrence of $X$ in $T$ is preceded by an
agreeing tier string.

A shortcut thus replaces $2^t \ell$ backward steps, each with
$O(\log r)$ time and poor locality, by $2^t \ell$ sequential character
comparisons and a constant number of rank and range-count queries. We
support counting only. Supporting locating would require maintaining a
row with known suffix-array value across shortcuts, as in the
$r$-index~\cite{GagieNavarroPrezza2020}, which we leave for future work.

\paragraph*{Pruning.}
Two rules remove shortcuts that would save little. First,
for an entrance and a tier $t$, let the \emph{area} of $t$ be the number
of rows with symbol greater than $t$ times $2^t \ell$, the number of
backward steps the tier saves if every such row is used once. Given a
threshold $\theta$, we keep tiers $0, 1, \ldots$ while their areas are at
least $\theta$ and cap the symbols at the first tier that falls short;
if tier $0$ falls short, the anchor is dropped. Second, patterns contain
only the characters of the collection's alphabet (\texttt{A}, \texttt{C},
\texttt{G}, \texttt{T} for genomes), so we drop an anchor whose window
contains another character (such as \texttt{N}) and cap an anchor's
tiers before the first tier string that contains one. We say that an
anchor removed by either rule is \emph{dropped} and that one whose tiers
are lowered is \emph{capped}; anchors removed because their windows
collide are \emph{discarded}.

\section{Construction}
\label{sec:construction}

We build the RLBWT first~\cite{BoucherEtAl2019}; it does not depend on
$\ell$, $\tau$ or $\theta$. Then we backward-search each window $\beta_p$,
giving its interval in $\ell$ steps, sort the anchors by interval and
discard those whose windows collide. For each remaining anchor $p$, we
continue the backward search for $\Rref[p - 2^{\tau_p} \ell .. p + \ell)$,
read from right to left, which gives the symbols of the rows of the
entrance $\intv(\beta_p)$. Let $J_k = \intv(\Rref[p - k .. p + \ell))$ be the
interval after $k$ further steps, and $Q_k$ the set of rows of the
entrance whose suffixes are preceded by $\Rref[p - k .. p)$. By the
argument of \Cref{lem:jump}, applying LF $k$ times maps $Q_k$ one to one
and in order onto $J_k$. In step $k + 1$, with $c = \Rref[p - k - 1]$, the rows
of $J_k$ with BWT character $c$ are mapped onto $J_{k+1}$ and the others
\emph{die}: the corresponding rows of $Q_k$ are not preceded by
$\Rref[p - k - 1 .. p)$. We keep $Q_k$ in an order-statistics
tree~\cite{Fenwick1994} over the entrance, list the runs intersecting
$J_k$, and for each run with a character other than $c$, from right to
left so that ranks are not disturbed, delete the elements of $Q_k$ at the
corresponding ranks. A row deleted in step $k +
1$ gets symbol $\sym = |\{t : 2^t \ell \le k\}|$; rows still alive after
$2^{\tau_p}\ell$ steps get $\tau_p + 1$. We also record $f_{p,t}$ as the
first row of $J_{2^t\ell}$. Since $\Rref$ occurs in $T$, $J_k$ and
$Q_k$ are never empty. The search stops early if the next character of
$\Rref$ is not in the alphabet, and $\tau_p$ is lowered accordingly.
Pruning by area uses the number of rows with each symbol and takes
$O(\tau + 1)$ time per anchor.

\begin{theorem}
\label{thm:construction}
Let $W_0 \ge W$ be the number of rows in the entrances of the anchors
that remain after discarding collisions, before pruning. Given the RLBWT
and $\Rref$, the shortcuts can be built in $O(|\Rref| (\tau + 1) \log r +
W_0 \log n + (|\Rref| / \ell) \log(|\Rref| / \ell))$ time, using
$O(|\Rref| / \ell + \max_p |\intv(\beta_p)|)$ words of working space in
addition to the RLBWT and the output.
\end{theorem}

\begin{proof}
Finding all windows' intervals takes $O(|\Rref|)$ backward steps, and
sorting the $O(|\Rref| / \ell)$ anchors takes $O((|\Rref| / \ell)
\log(|\Rref| / \ell))$ time. A remaining anchor with $\tau_p = t$ needs
$2^t\ell$ further backward steps. At
most $|\Rref| / (2^{t+1}\ell) + 1$ anchors have $\tau_p = t < \tau$, and
at most $|\Rref| / (2^\tau \ell) + 1$ have $\tau_p = \tau$, so there are
$O(|\Rref| (\tau + 1))$ backward steps in total, each taking $O(\log r)$
time. In step $k + 1$, every run boundary inside $J_k$ is adjacent to a
run whose character is not $c$, since adjacent runs have different
characters, and every such run contains at least one dying row. So the
number of runs listed is at most one plus twice the number of rows dying
in that step, and over all steps the listing costs $O(\log r)$ per step
plus $O(1)$ per dying row. Each row of an entrance is deleted at most
once, at $O(\log n)$ cost, which gives $O(W_0 \log n)$ in total.
\end{proof}

On a repetitive collection, the rows of an entrance come mostly from
copies of the same region of $\Rref$, and they die only where the copies
differ from the reference, so the work beyond the backward steps is
proportional to the number of such differences, not to the length of
the shortcuts.

\section{Experiments}
\label{sec:experiments}

We tested backward search with shortcuts on two collections: the 37
\texttt{cere} genomes (\Cref{sec:cere}) and an 11-gigabase collection of
genomes of seven bacterial species, with one chromosome per species as
reference (\Cref{sec:bacteria}). We then used shortcuts to find maximal
exact matches in whole human chromosomes (\Cref{sec:chr19}), with the
setup described there. The code and scripts for all three are in the
directory \texttt{chr19/shortcuts} of
\url{https://github.com/TravisGagie/MEM-finding}. On the first two collections, we used $\tau = 3$, so shortcuts of
$\ell$, $2\ell$, $4\ell$ and $8\ell$ characters, and built the
shortcuts for several values of $\ell$ and of $\theta$. We took 200
substrings of the collection of each length in
$\{200, 400, 600, 800, 1000\}$, starting at uniformly random positions
and consisting only of \texttt{A}, \texttt{C}, \texttt{G} and
\texttt{T}, as exact patterns, and made further pattern sets by
substituting each of their characters with a fixed probability. For each
configuration, pattern set and length, we searched for the 200 patterns
with and without shortcuts, timing each search three times and keeping
the fastest, and summed the times; the searches with and without
shortcuts alternate, so both see the same conditions. A speedup is the
ratio of the two sums for the same configuration; for each pattern set
and length, the sums for plain search differed from their median over
the configurations by less than 7\% on \texttt{cere} and less than 3\%
on the bacterial one. We
also measured the size of the RLBWT plus shortcuts and, in a separate
untimed run of each search, the report.

Our RLBWT stores the run starts in two Elias--Fano
bitvectors~\cite{OkanoharaSadakane2007}, in BWT order and with the runs
sorted by character, and the run characters in a wavelet tree; plain
bitvectors, including the high parts of the Elias--Fano ones, store rank
counts per 512-bit block and per 64-bit word, as in Vigna's
rank9~\cite{Vigna2008}. The shortcuts are as in \Cref{sec:design}, except
that $\Rref$ is stored with one byte per character and the tier data in
plain arrays. The experiments on the first two collections ran on a laptop with a 13th Gen Intel Core
i5-1335U processor and 16\,GB of RAM, under Linux, with the searches
pinned to one performance core. In both experiments, every
search with shortcuts found the same longest matching suffix as plain
backward search and, for the patterns that occur in full, the same
interval.

\subsection{The \texttt{cere} genomes}
\label{sec:cere}

\paragraph*{Setup.}
We indexed the Pizza\&Chili \texttt{cere} collection~\cite{FerraginaEtAl2009},
37 genomes of \emph{Saccharomyces cerevisiae} with $n = 461{,}286{,}646$
characters in total, taking as reference $\Rref$ its first
$12{,}000{,}000$ characters, roughly the first genome. We built the RLBWT
once, from a suffix array computed with libsais~\cite{Grebnov2021}, and
then built the shortcuts for $\ell \in
\{25, 50, 100, 200\}$ and $\theta \in \{0, 100, \ldots, 1000\}$: 44
configurations in all. Besides the exact patterns, we used a pattern set
with 1\% of the characters substituted; this makes $88{,}000$ searches
with shortcuts.

\paragraph*{Construction.}
Building the RLBWT took $27$ seconds, with the text and its 64-bit suffix
array in memory (9 bytes per character); it has $r = 11{,}574{,}642$
runs ($n/r \approx 39.9$) and takes $50.1$\,MiB. Building the shortcuts
for one value of $\ell$, before pruning, took $17$ seconds:
about $38.5$ million backward steps, slightly fewer than the
$3.5\,|\Rref| = 42$ million that the proof of \Cref{thm:construction}
gives for $\tau = 3$, because discarded and dropped anchors are not
extended and capped ones are extended less, and between $0.5$ million
($\ell = 200$) and $2.7$ million ($\ell = 25$) deaths. Pruning for each
value of $\theta$ then took at most $0.23$ seconds. Of the $479{,}999$
anchors for $\ell = 25$, $37{,}421$ were dropped because their window or
tier-$0$ string contains a character other than \texttt{A}, \texttt{C},
\texttt{G} or \texttt{T} (the reference contains runs of \texttt{N}),
and $1{,}183$ were discarded because their windows collide; for
$\ell = 50$ these numbers are $19{,}074$ and $76$ of $239{,}999$, and for
$\ell = 200$ no windows collide.

\begin{figure}[t]
  \centering
  \includegraphics[width=\textwidth]{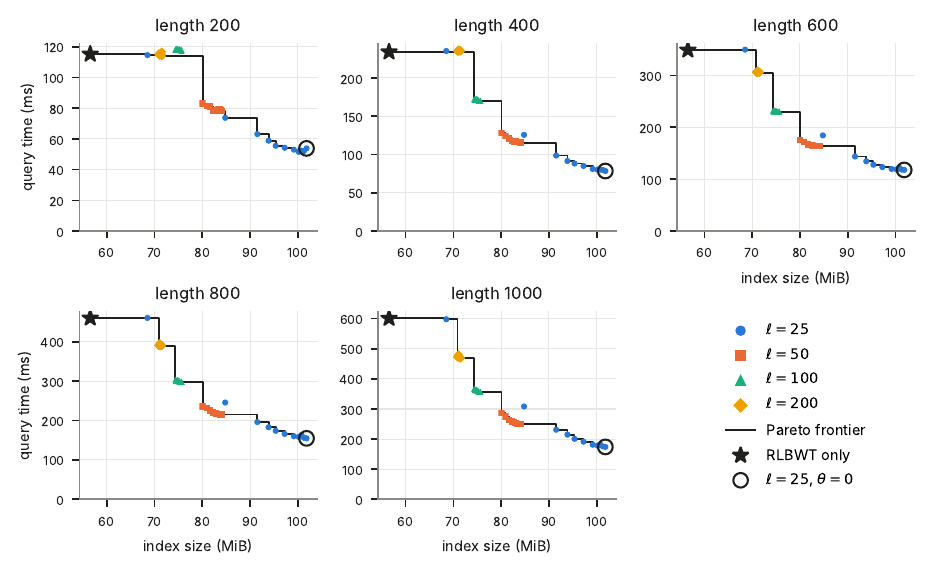}
  \caption{Total index size (RLBWT plus shortcuts) versus total time to
    search for 200 exact patterns, one panel per pattern length, on
    \texttt{cere}. Each point is one configuration $(\ell, \theta)$; the
    line is the Pareto frontier; the star is the RLBWT alone, at the
    median time of plain search over all configurations; the circle is
    $\ell = 25$, $\theta = 0$.}
  \label{fig:pareto}
\end{figure}

\begin{table}[t]
  \centering
  \caption{Index size and speedup over plain backward search for
    selected configurations, on exact patterns and on patterns with 1\%
    of their characters substituted, and the number of shortcuts per
    exact pattern of length 1000 together with the fraction of matched
    characters they cross.}
  \label{tab:results}
  \begin{tabular}{rr r ccc ccc cc}
    \toprule
    & & & \multicolumn{3}{c}{exact} & \multicolumn{3}{c}{1\% substituted} & \multicolumn{2}{c}{report, length 1000} \\
    \cmidrule(lr){4-6} \cmidrule(lr){7-9} \cmidrule(lr){10-11}
    $\ell$ & $\theta$ & MiB & 200 & 600 & 1000 & 200 & 600 & 1000 & shortcuts & crossed \\
    \midrule
    \multicolumn{2}{l}{RLBWT} & 50.1 & 1.00 & 1.00 & 1.00 & 1.00 & 1.00 & 1.00 & -- & -- \\
    25 & 0 & 93.2 & 2.32 & 3.05 & 4.17 & 1.70 & 1.81 & 1.95 & 9.3 & 79\% \\
    25 & 800 & 83.3 & 1.88 & 2.40 & 2.96 & 1.52 & 1.55 & 1.71 & 9.1 & 70\% \\
    50 & 0 & 77.1 & 1.46 & 2.02 & 2.61 & 1.17 & 1.26 & 1.28 & 5.1 & 65\% \\
    100 & 0 & 69.0 & 0.95 & 1.39 & 1.70 & 0.98 & 1.00 & 1.04 & 2.3 & 45\% \\
    200 & 0 & 65.1 & 1.01 & 1.06 & 1.22 & 1.01 & 1.00 & 1.01 & 0.9 & 24\% \\
    \bottomrule
  \end{tabular}
\end{table}

\paragraph*{Results.}
\Cref{fig:pareto} and \Cref{tab:results} show the trade-off between size
and speed. The window length $\ell$ is the main parameter: halving it
doubles the number of entrances, increases the space for shortcuts (from
$15.0$\,MiB for $\ell = 200$ to $43.1$\,MiB for $\ell = 25$), and makes
searches faster, since a search reaches an entrance sooner and more
often. With exact patterns, $\ell = 25$ and $\theta = 0$ make searches
$2.3$ to $4.2$ times faster for 86\% more space, and $\ell = 50$ makes
them $1.5$ to $2.6$ times faster for 54\% more space; the speedup
grows with the pattern length, since a longer match crosses more
windows. With $\ell = 200$, no search for a pattern of length 200 took a
shortcut, and with $\ell = 100$ few did; checking for entrances, one rank
query per step, made these searches at most 5\% slower.

Pruning with $\theta \le 400$ changed the speedups by at most 5\% and
saved at most $2.7$\,MiB. Larger thresholds matter for $\ell = 25$: an
anchor is dropped if fewer than $\theta / \ell$ rows of its entrance
survive tier~0, and since there are 37 genomes, most entrances have
fewer than 40 rows. With $\theta = 800$, $296{,}946$ of the $441{,}395$
entrances remain, saving $10.0$\,MiB but lowering the speedup at length
1000 from $4.17$ to $2.96$; with $\theta = 1000$, only $4{,}574$ remain,
almost no shortcuts are taken, and checking for entrances makes searches
4 to 8\% slower than without shortcuts.

With 1\% of the characters substituted, only 26 of the 1000 patterns
occur in full (23 of length 200 and 3 of length 400), and the longest
matching suffix has 91 characters on average, but $\ell = 25$ still
makes searches $1.7$ to $2.0$ times faster.

With $\ell = 25$, a search for an exact pattern of length 1000 reports
$9.3$ shortcuts on average, each at a different anchor, which cross 79\%
of the matched characters; each certifies an exact alignment of a piece
of the pattern to $\Rref$ (\Cref{sec:linear}). With $\ell = 50$ the
report has $5.1$ shortcuts, which cross 65\% of the matched characters.

\subsection{Bacterial genomes}
\label{sec:bacteria}

\paragraph*{Setup.}
For a larger and more varied collection, we used the reference sequences
of SPUMONI~2~\cite{AhmedEtAl2023}: RefSeq sequences of seven bacterial
species, those of the ZymoBIOMICS mock community. For each species, in
the order of its file, we took the first record of at least 2 million
bases, a complete chromosome, as part of the reference, and then further
records until the next would bring the species past 2 billion bases.
The reference chromosomes, between $2.9$ and $6.8$ million bases long,
are those of \emph{B.~subtilis} H19, \emph{E.~faecalis} 28157\_4\#332,
\emph{E.~coli} K1508, \emph{L.~monocytogenes} UKVDL7,
\emph{P.~aeruginosa} SRRSH1101, \emph{S.~enterica} serovar Typhimurium
BBS1359 and \emph{S.~aureus} USA300-SUR3; we took all the records of
\emph{B.~subtilis}, \emph{E.~faecalis} and \emph{L.~monocytogenes}, and
about 2 billion bases of each of the other four species. The collection has
$5{,}709$ records, $2{,}736$ of them at least 2 million bases long, and
$11{,}043{,}326{,}242$ bases, all on the forward strand; the reference
$\Rref$ is the seven chromosomes, $29{,}004{,}978$ bases separated by
record separators, which are not \texttt{A}, \texttt{C}, \texttt{G}
or \texttt{T}, so no window or tier string crosses from one species to
another. To save memory, we built the BWT with
Big-BWT~\cite{BoucherEtAl2019}, an implementation of prefix-free
parsing, and read it into our RLBWT. We built the
shortcuts for $\ell \in \{16, 25, 50, 100\}$ and
$\theta \in \{0, 400, 800\}$, and used three pattern sets:
exact, with 0.1\% of the characters substituted, and with 1\%
substituted; this makes $36{,}000$ searches with shortcuts.

\paragraph*{Construction.}
Big-BWT took 17 minutes with a peak memory use of $11.2$\,GiB. The RLBWT
has $n = 11{,}043{,}331{,}952$ characters and $r = 238{,}635{,}263$ runs
($n/r \approx 46.3$) and takes $1{,}048.3$\,MiB; reading the BWT into it
took 55 seconds. Building the shortcuts
for one value of $\ell$ took 113 to 123 seconds: about $101$ million
backward steps, close to $3.5\,|\Rref|$, and between $16.3$ million
($\ell = 100$) and $119.3$ million ($\ell = 16$) deaths. Of the
$1{,}812{,}810$ anchors for $\ell = 16$, $11{,}193$ were discarded
because their windows collide; for $\ell = 25$, $2{,}106$ of
$1{,}160{,}198$ were.

\begin{table}[t]
  \centering
  \caption{Index size and speedup over plain backward search on the
    bacterial collection with $\theta = 0$, on exact patterns and on
    patterns with 1\% of their characters substituted, and the number of
    shortcuts per exact pattern of length 1000 together with the fraction
    of matched characters they cross.}
  \label{tab:bacteria}
  \begin{tabular}{r r ccc ccc cc}
    \toprule
    & & \multicolumn{3}{c}{exact} & \multicolumn{3}{c}{1\% substituted} & \multicolumn{2}{c}{report, length 1000} \\
    \cmidrule(lr){3-5} \cmidrule(lr){6-8} \cmidrule(lr){9-10}
    $\ell$ & MiB & 200 & 600 & 1000 & 200 & 600 & 1000 & shortcuts & crossed \\
    \midrule
    RLBWT & $1{,}048.3$ & 1.00 & 1.00 & 1.00 & 1.00 & 1.00 & 1.00 & -- & -- \\
    16 & $1{,}380.5$ & 1.58 & 1.58 & 1.88 & 1.37 & 1.27 & 1.59 & 9.8 & 50\% \\
    25 & $1{,}254.3$ & 1.39 & 1.44 & 1.75 & 1.21 & 1.17 & 1.45 & 6.4 & 44\% \\
    50 & $1{,}155.8$ & 1.14 & 1.25 & 1.50 & 1.04 & 1.01 & 1.24 & 3.0 & 33\% \\
    100 & $1{,}110.2$ & 0.98 & 1.07 & 1.26 & 0.99 & 0.99 & 1.08 & 1.2 & 20\% \\
    \bottomrule
  \end{tabular}
\end{table}

\paragraph*{Results.}
\Cref{tab:bacteria} shows the results. Shortcuts help less than on
\texttt{cere}, and shorter windows help more: $\ell = 16$, which we did
not try on \texttt{cere}, is the fastest, making searches for exact
patterns $1.6$ to $1.9$ times faster for 32\% more space, and
$\ell = 25$ makes them $1.4$ to $1.75$ times faster for 20\% more space.
The pruning threshold matters little here: going from $\theta = 0$ to
$\theta = 800$ changes the speedups by at most $0.03$ and saves
$28$\,MiB for $\ell = 16$.

We think the main reason is variation between strains: almost every
pattern comes from a genome other than the reference, and it agrees
with $\Rref$ only between the places where its strain differs from the
reference strain, while a shortcut needs $2\ell$ agreeing characters
lined up with an anchor. With $\ell = 16$,
shortcuts cross only half of the matched characters of an exact pattern
of length 1000, against 79\% on \texttt{cere} with $\ell = 25$. Patterns
with 0.1\% of their characters substituted behave almost like exact
ones, with speedups of $1.55$ to $1.9$ for $\ell = 16$, because the
differences between strains outweigh the substitutions. With 1\%
substituted, only 33 of the 1000 patterns occur in full and the longest
matching suffix has 96 characters on average, and $\ell = 16$
still makes searches $1.3$ to $1.7$ times faster. How much shortcuts
help thus depends on how close the collection is to the reference, which
suggests using several references per species, or a pangenome graph.

\subsection{Finding MEMs in whole chromosomes}
\label{sec:chr19}

Shortcuts help most when matches are long. To find maximal exact
matches (MEMs) we need searches in both directions, so we index the
collection $T$ and the reversed collection $T'$, each with an RLBWT and
shortcuts, and place the windows at the same boundaries of the
reference in both, so that forward and backward searches both take
shortcuts. We use Boyer--Moore--Li (BML)~\cite{Gagie2024BML,Li2024},
which finds only the MEMs of length at least a threshold $L$, and
forward--backward search~\cite{Li2012}, which finds all of them; both
use only searches for the longest suffix of a substring of the pattern
that occurs in the text, so shortcuts change their running time but not
their output.

\paragraph*{Setup.}
From release~2 of the Human Pangenome Reference
Consortium~\cite{HPRC2025,LucasEtAl2026} we took chromosome~19. The
collections consist of the reference CHM13~\cite{NurkEtAl2022}, of 61.7
million bases, and both haplotypes of the first 16, 32 or 64
individuals, by sample identifier, whose haplotypes each have
chromosome~19 as a single sequence of at least 90\% of the length of
CHM13; with 16 individuals there are 33 sequences and
$n = 2{,}032{,}742{,}766$ characters. The patterns are haplotype~1 of
the next two such individuals, HG00423 and HG00438, of 61.7 and 61.0
million bases, split at runs of \texttt{N} into 4 patterns of 122.4
million characters in all. All sequences were already in the
orientation of CHM13, so we index the forward strands only. The
reference $\Rref$ is CHM13, with $\ell = 25$ and $\tau = 6$; other values
were slower, since shorter windows collide more often and longer ones
are crossed less often. We built the BWTs with
Big-BWT~\cite{BoucherEtAl2019} and used the same RLBWT as in the other
experiments. These runs were on a desktop with an Intel Core i7-8700
CPU (3.2\,GHz, 6 cores) and 16\,GiB of RAM, with the searches pinned to
one core and a single thread for every tool; we report the median of
five runs. With and without shortcuts, our searches gave the same MEMs,
with the same intervals. For comparison we also ran
ropebwt3~\cite{Li2024} and Movi~2~\cite{ZakeriEtAl2026}, single-threaded.
Both index both strands and report a MEM if it or its reverse complement
matches, so they answer a slightly more general question; we checked
that the differences between ropebwt3's output and ours are
reverse-strand matches, and most of Movi~2's differences are explained
the same way.

\begin{table}[t]
  \centering
  \caption{Finding the MEMs of two chromosome~19 haplotypes (122.4
    million characters) with respect to 16, 32 and 64 individuals: space
    in MiB (for both directions, for our indexes), and time in seconds,
    single-threaded, for BML with $L = 1000$ and $L = 100$ and for
    forward--backward search (FB) with $L = 1000$. Plain search uses only
    the two RLBWTs. Ropebwt3 and Movi~2 index both strands; for Movi~2
    the FB column gives the time to find all MEMs.}
  \label{tab:chr19}
  \begin{tabular}{@{}llrrrr@{}}
    \toprule
    individuals & method & MiB & BML, 1000 & BML, 100 & FB, 1000 \\
    \midrule
    16 & plain search & 314 & 290.5 & 287.3 & 287.2 \\
       & shortcuts    & 543 & 37.2 & 39.5 & 39.3 \\
       & ropebwt3     & 157 & 64.3 & 65.5 & 63.1 \\
       & Movi~2       & 554 & 25.9 & 28.9 & 57.3 \\
    \midrule
    32 & plain search & 349 & 321.5 & 318.2 & 318.6 \\
       & shortcuts    & 699 & 40.3 & 42.2 & 41.5 \\
       & ropebwt3     & 178 & 72.8 & 73.9 & 69.0 \\
       & Movi~2       & 586 & 30.2 & 34.2 & 65.1 \\
    \midrule
    64 & plain search & 406 & 348.7 & 345.5 & 344.3 \\
       & shortcuts    & 998 & 43.1 & 44.8 & 44.2 \\
       & ropebwt3     & 212 & 79.5 & 80.2 & 76.8 \\
       & Movi~2       & 639 & 35.5 & 40.8 & 85.1 \\
    \bottomrule
  \end{tabular}
\end{table}

\paragraph*{Results.}
With BML and $L = 1000$, plain search takes about 3.7 backward steps
per pattern character, since a forward search and two backward searches
cover each MEM and consecutive MEMs overlap. Shortcuts reduce this to
0.45 and make BML 7.8 times faster with 16 individuals, for 1.7 times
the space of the two RLBWTs (\Cref{tab:chr19}). The speedup stays about
the same, 8.0 times, as the collection grows to 64 individuals, but the
space grows to 2.5 times that of the RLBWTs, mostly for the per-row tier
symbols: new haplotypes add rows to the entrances without adding
shortcuts. With shortcuts, our search is 1.7 to 1.8 times faster than
single-threaded ropebwt3, whose index takes about a quarter of the
space, but 1.2 to 1.4 times slower than Movi~2, whose index takes 0.6 to
1.0 times the space and uses a move structure instead of rank queries;
both index both strands, so per strand our index is larger. As with the
bacterial collection, the speedup is limited by how much the collection
differs from the reference: a shortcut is taken only where the pattern
matches the reference for the whole length of the shortcut.

\subsection{Speed of the baseline}

These speedups are relative to our own RLBWT, in which a backward step
takes, depending on the pattern length, $0.13$ to $0.23\,\mu$s on
\texttt{cere} and $0.19$ to $0.36\,\mu$s on the bacterial collection,
whose RLBWT is about 20 times larger. Indexes
with faster LF, such as move structures~\cite{NishimotoTabei2021} as
implemented in Movi~\cite{ZakeriEtAl2024}, may take less time per step,
so each step a shortcut avoids would save less time. On
chromosome~19, BML with shortcuts was 1.2 to 1.4 times slower than
Movi~2~\cite{ZakeriEtAl2026}, which has no shortcuts
(\Cref{sec:chr19}); we have not measured how much of the speedup
remains when shortcuts are added to a move structure. Suffixient
arrays~\cite{CenzatoEtAl2024} and indexes based on suffix-tree path
decompositions~\cite{BeckerEtAl2026} are faster still, but they do not
give BWT intervals.

\end{document}